\documentclass[aps,twocolumn,prd,10pt,preprintnumbers,superscriptaddress]{revtex4-2}
\usepackage{physics}
\usepackage{amsmath, amssymb, graphics, setspace}
\usepackage{amsthm}
\newtheorem{theorem}{Theorem}
\newtheorem{corollary}[theorem]{Corollary}
\usepackage{tikz}
\tikzset{every picture/.style={line width=0.75pt}}
\usepackage{tabularray}
\usepackage{orcidlink}
\usepackage{hyperref}

\begin{document}
\preprint{USTC-ICTS/PCFT-26-27}
\title{Genus drop involving non-hyperelliptic curves in Feynman integrals}

\author{Feiyu Yang\,\orcidlink{0009-0009-9304-5657}}
\email{fyyang22@mail.ustc.edu.cn}
\affiliation{Interdisciplinary Center for Theoretical Study, University of Science and Technology of China, Hefei, Anhui 230026, China}

\author{Jianyu Gong\,\orcidlink{0000-0002-3325-5777}}
\email{jianyu\_gong@sjtu.edu.cn}
\affiliation{Institute of Nuclear and Particle Physics (INPAC), Shanghai Jiao Tong University, Shanghai 200240, China}

\author{Yang Zhang\,\orcidlink{0000-0001-9151-8486}}
\email{yzhphy@ustc.edu.cn}	
\affiliation{Interdisciplinary Center for Theoretical Study, University of Science and Technology of China, Hefei, Anhui 230026, China}
\affiliation{Peng Huanwu Center for Fundamental Theory, Hefei, Anhui 230026, China}
\affiliation{Center for High Energy Physics, Peking University, Beijing 100871, People’s Republic of China}

\begin{abstract}
For both theoretical and phenomenological studies, it is important to analyze the function types of Feynman integrals. The phenomenon of genus drop between different representations of hyperelliptic Feynman integrals was discussed in \cite{Marzucca2024Genusdrop}. In this paper, 
we reformulate the extra-involution mechanism of \cite{Marzucca2024Genusdrop} as a special case of an unramified double covering between algebraic curves, and show that this covering mechanism also explains genus drops accompanied by a curve-type change from non-hyperelliptic to hyperelliptic for a class of three-loop Feynman diagrams. We also demonstrate that within a specific framework, the origin of the discrete spacetime symmetry that leads to the genus drop in hyperelliptic cases is manifest. This work also points out that there exist non-hyperelliptic Feynman integrals that exhibit no apparent genus drop.
\end{abstract}

\maketitle

\section{Introduction}
Precise computation of Feynman integrals is crucial for the theoretical input in modern particle physics. To compute Feynman integrals efficiently, it is very helpful to have a deep understanding of the mathematical structure of Feynman integrals. A major complexity of Feynman integrals comes from the various types of special functions in Feynman integrals. In one-loop cases, there are multiple polylogarithms \cite{Chen1977Iteratedpath,Goncharov1998Multiplepolylogarithms,Goncharov2001Multiplepolylogarithms,Duhr2019FunctionTheory,Duhr2025Analyticresults}, which are iterated integrals of rational functions with simple poles. Starting from two-loop, more non-trivial integration kernels would appear, especially when there are massive propagators. Usually, these integration kernels can be written as integrations of rational functions defined on some projective algebraic varieties, like elliptic curves \cite{Bauberger1995Analyticalnumerical,Laporta2005Analytictreatment,Muller-Stach2012secondorderdifferential,Adams2013twoloopsunrise,Muller-Stach2014PicardFuchsequations,Adams2014twoloopsunrise,Bloch2015ellipticdilogarithm,Adams2015twoloopsunrise,Adams2016walksunset,Adams2016iteratedstructure,Adams2016sunriseintegral,Adams2016kiteintegral,Bogner2017Analyticcontinuation,Bourjaily2018EllipticDoubleBox,Broedel2018Ellipticpolylogarithmsa,Adams2018Feynmanintegrals,Adams2018eformdifferential,Broedel2018Ellipticpolylogarithms,Adams2018planardouble,Adams2018Analyticresults,Honemann2018electronselfenergy,Broedel2019EllipticFeynman,Bogner2020unequalmass,Weinzierl2021Modulartransformations,Frellesvig2022EpsilonFactorized,Muller2022Feynmanintegral,Dlapa2023algorithmicapproach,Delto2024TwoLoopQED,Badger2024Twoloopintegrals,Frellesvig2024efactorisedbases,Duhr2024electronselfenergy,Marzucca2025TwoLoopMaster,Becchetti2025Analytictwoloop}, higher-genus curves \cite{Huang2013GeneraCurves,Hauenstein2015GlobalStructure,Georgoudis2015TwoloopIntegral,Duhr2025Canonicaldifferential}, Calabi--Yau manifolds \cite{Brown2012K3phi4,Bloch2015Feynmanintegral,Bloch2018Localmirror,Bourjaily2018TraintracksCalabiYaus,Bourjaily2019BoundedBestiary,Bourjaily2020EmbeddingFeynman,Klemm2020lloopBanana,Bonisch2021AnalyticStructure,Bourjaily2022FunctionsMultiple,Bonisch2022Feynmanintegrals,Pogel2023TamingCalabiYau,Lairez2023Algorithmsminimal,Frellesvig2024CalabiYaumeets,Frellesvig2024ClassifyingpostMinkowskian,Frellesvig2024CalabiYauFeynman,Duhr2025Aspectscanonical,Duhr2025Threeloopbanana,Pogel2025unequalmassthreeloop}, and so on \cite{Schimmrigk2024SpecialFano,Cruz2025FanoReflexive,Bargiela2025spectrumFeynmanintegral}. 

To determine the geometry of a Feynman integral, a common method is to take the maximal cut, where all propagators are on shell. Maximal cuts are particularly suitable for studying the geometric structure of Feynman integrals, not only because they are relatively easy to calculate, but also because it provides a solution of the homogeneous part of the differential equations satisfied by Feynman integrals \cite{Primo2017maximalcut}. However, it is illustrated in \cite{Marzucca2024Genusdrop} that the maximal cuts taken in different representations can give algebraic curves of different genera; the simpler, lower-genus curve alone is sufficient for the computation of Feynman integrals \cite{Duhr2025Canonicaldifferential}. This phenomenon is particularly intriguing because it reveals a gap in our understanding of the geometric structures of Feynman integrals: while simplifications of the geometric structures are possible, it is still not fully guaranteed to find the minimal geometric structure for a high-loop Feynman integral. Therefore, it is necessary to conduct an in-depth study of the genus drop phenomenon to understand its underlying physical origin.

This paper investigates two types of genus drop:
\begin{description}
\item[Type I] genus drops accompanied by a curve-type change from non-hyperelliptic to hyperelliptic;
\item[Type II] genus drops between hyperelliptic curves.
\end{description}
Among these, Type I is new and has not been analyzed in the literature: here the first curve is non-hyperelliptic, so the method of \cite{Marzucca2024Genusdrop} does not directly apply. 
A hyperelliptic curve is a smooth projective algebraic curve of genus $g\geq 2$ that admits an affine model of the form $y^2=f(x)$, where $f(x)$ is a polynomial of degree $n$ with no multiple roots. Its genus is given by $\lfloor \frac{n-1}{2} \rfloor$. Curves of genus $g \geq 3$ that do not admit such a form are called non-hyperelliptic (a genus-$2$ curve always admits such a form \cite{miranda1995algebraic}). Non-hyperelliptic curves are generally more complicated to handle, as they lack a simple uniform equation as in the hyperelliptic case.

This paper analyzes the mechanisms behind both types of genus drop from mathematical and physical viewpoints. Mathematically, we demonstrate that both types are induced by a specific kind of map between algebraic curves: the unramified double covering. Within this framework, the genus drops follow the pattern
\begin{equation}
    g\mapsto \frac{g+1}{2},
\end{equation}
as dictated by the Riemann--Hurwitz formula. The property of being unramified also tells us how the period matrices of the curves of different genera are related, via Prym varieties \cite{Mumford1974PrymVarieties}. Physically, Type I  is tied to the way the maximal cut is applied; Type II is tied to the symmetry of the spacetime. We will show that the spacetime dimension spanned by the external legs induces a spacetime reflection symmetry in the Feynman integral, which leads to the genus drop in the purely hyperelliptic case and confirms the expectation of \cite{Marzucca2024Genusdrop} that the genus drop is induced by discrete Lorentz symmetries.

To illustrate our idea, we give one detailed example for each case and explicitly derive the unramified double covering maps. In the Type I case, we take the box-pentagon-box diagram with massive propagators as our main example. This diagram has three loops and eleven massive propagators and is therefore expected to be associated with a one-dimensional algebraic variety under maximal cuts in four-dimensional spacetime, as demonstrated in \cite{Huang2013GeneraCurves,Hauenstein2015GlobalStructure}. In the Type II case, we take the non-planar crossed box as in \cite{Marzucca2024Genusdrop} to show how the reflection symmetry of spacetime leads to the genus drop. Additionally, a genus criterion for determining whether an algebraic curve is hyperelliptic is provided in the Appendix. Finally, we note that there exists a non-hyperelliptic curve in a three-loop Feynman integral that does not follow any of the known genus-drop patterns.

\section{Type I}
Take the box-pentagon-box diagram with massless external legs and equal-mass propagators as an example to demonstrate how the genus drop happens.
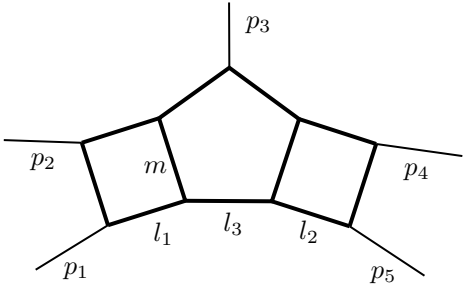
\begin{figure}[h]
\begin{tikzpicture}[x=0.75pt,y=0.75pt,yscale=-1,xscale=1]
\draw  [line width=1.5]  (278.45,74.79) -- (264.79,116.02) -- (221.36,115.77) -- (208.18,74.39) -- (243.46,49.06) -- cycle ;
\draw    (130.33,85.33) -- (169.46,86.68) ;
\draw    (182.59,128.02) -- (146.33,150.33) ;
\draw    (317.1,87.28) -- (360.33,93.33) ;
\draw    (303.77,128.55) -- (341.33,153.33) ;
\draw    (243.46,49.06) -- (243.33,16.33) ;
\draw [line width=1.5]    (169.46,86.68) -- (182.59,128.02) ;
\draw [line width=1.5]    (169.46,86.68) -- (208.18,74.39) ;
\draw [line width=1.5]    (317.1,87.28) -- (303.77,128.55) ;
\draw [line width=1.5]    (317.1,87.28) -- (278.45,74.79) ;
\draw [line width=1.5]    (303.77,128.55) -- (264.79,116.02) ;
\draw [line width=1.5]    (182.59,128.02) -- (221.36,115.77) ;
\draw (203.97,125.29) node [anchor=north west][inner sep=0.75pt]    {$l_{1}$};
\draw (277,124.4) node [anchor=north west][inner sep=0.75pt]    {$l_{2}$};
\draw (239,120.4) node [anchor=north west][inner sep=0.75pt]    {$l_{3}$};
\draw (159,145.4) node [anchor=north west][inner sep=0.75pt]    {$p_{1}$};
\draw (142.33,90.73) node [anchor=north west][inner sep=0.75pt]    {$p_{2}$};
\draw (250.33,21.73) node [anchor=north west][inner sep=0.75pt]    {$p_{3}$};
\draw (329.72,95.71) node [anchor=north west][inner sep=0.75pt]    {$p_{4}$};
\draw (313,147.4) node [anchor=north west][inner sep=0.75pt]    {$p_{5}$};
\draw (199,94.4) node [anchor=north west][inner sep=0.75pt]    {$m$};
\end{tikzpicture}
\caption{The box-pentagon-box diagram with massive internal lines}
\label{fig:bpb}
\end{figure}

The propagators are given by
\begin{equation}
    \begin{aligned}[b]
        &D_1=l_1^2-m^2,\quad D_2=(l_1-p_1)^2-m^2,\\
        &D_3=(l_1-p_1-p_2)^2-m^2,\quad D_4=l_2^2-m^2,\\
        &D_5=(l_2-p_5)^2-m^2,\quad D_6=(l_2-p_5-p_4)^2-m^2,\\
        &D_7=l_3^2-m^2,\quad D_8=(l_3-p_1-p_2)^2-m^2,\\
        &D_9=(l_3-p_1-p_2-p_3)^2-m^2,\\
        &D_{10}=(l_1-l_3)^2-m^2,\quad D_{11}=(l_2+l_3)^2-m^2.
    \end{aligned}
\end{equation}
Consider its four-dimensional geometry in both the momentum representation and the loop-by-loop Baikov representation \cite{Frellesvig2017CutsFeynman}.

\subsection{Momentum representation}
In the momentum representation, use the method in \cite{Huang2013GeneraCurves}, where the associated algebraic curve is the variety defined by the $4L-1$ cut equations in $\mathbb{CP}^{4L}$, where $L$ is the loop number. The spinor helicity parameterization is
\begin{align}
    &l_1=a_{1,1}p_1+a_{1,2}p_2+\frac{a_{1,3}}{2} \langle p_1|\gamma|p_2]+\frac{a_{1,4}}{2} \langle p_2|\gamma|p_1],\\
    &l_2=a_{2,1}p_5+a_{2,2}p_4+\frac{a_{2,3}}{2} \langle p_5|\gamma|p_4]+\frac{a_{2,4}}{2} \langle p_4|\gamma|p_5],\\
    &l_3=a_{3,1}K_1+a_{3,2}K_2+\frac{a_{3,3}}{2} \langle K_1|\gamma|K_2]+\frac{a_{3,4}}{2} \langle K_2|\gamma|K_1],\label{eq:l3}
\end{align}
where
\begin{align}
    &K_1=p_3,\\
    &K_2=p_1+p_2-\frac{s_{12}}{s_{13}+s_{23}}p_3.
\end{align}
This parameterization is inspired by \cite{Badger2012HeptaCutsTwoLoop}. When the external legs attached to a loop are not massless (in this case $p_1+p_2$ can be treated as an external leg attached to the loop $l_3$), the definitions of $K_1$ and $K_2$ keep the cut solution as simple as in the case with massless external legs.

Substituting the parameterization into the first nine cut equations $\{D_1=0,\dots,D_9=0\}$, the solution is as follows:
\begin{align}
    &l_1=p_1+\frac{t_1}{2}\langle p_1|\gamma|p_2]-\frac{1}{t_1}\frac{m^2}{s_{12}}\langle p_2|\gamma|p_1],\\
    &l_2=p_5+\frac{t_2}{2}\langle p_5|\gamma|p_4]-\frac{1}{t_2}\frac{m^2}{s_{45}}\langle p_4|\gamma|p_5],\\
    &l_3=K_2+\frac{t_3}{2}\langle K_1|\gamma|K_2]-\frac{1}{t_3}\frac{m^2}{s_{13}+s_{23}}\langle K_2|\gamma|K_1],
\end{align}
where we have renamed $a_{1,3}, a_{2,3}, a_{3,3}$ to $t_1,t_2,t_3$ for convenience. Then, the algebraic curve is defined by the intersection of the hypersurfaces defined by the common zeros of the remaining cut propagators $D_{10}(t_1,t_3)$ and $D_{11}(t_2,t_3)$. Use $f_1(t_1,t_3)$ and $f_2(t_2,t_3)$ to denote the numerators of $D_{10}(t_1,t_3)$ and $D_{11}(t_2,t_3)$ respectively, then the curve is defined by
\begin{equation}\label{eq:mom}
    \begin{cases}
        \,f_1(t_1,t_3)=0,\\
        \,f_2(t_2,t_3)=0.
    \end{cases}
\end{equation}
The curve is denoted as $\mathcal{C}_\text{mom.}=V(f_1,f_2)$. In what follows, by a curve we always mean the smooth projective normalization (compact Riemann surface) of the corresponding affine algebraic curve.

We use the Riemann--Hurwitz formula (Theorem~\ref{RH} in the Appendix) to find its genus. First, note that both $f_1=0$ and $f_2=0$ define elliptic curves in their respective subspaces.  Calculating the discriminants of $f_1$ and $f_2$ (both are quadratic polynomials) with respect to any variable, it turns out that all the discriminants are degree-$4$ polynomials with no multiple roots. Since $f_2$ is quadratic in $t_2$, for each $(t_1,t_3)\in V(f_1)$, there exist two points $\{(t_1,t_2^{(+)},t_3), (t_1,t_2^{(-)},t_3)\} \in V(f_1,f_2)$ except at a finite set of points (where the discriminant of $f_2$ with respect to $t_2$ becomes zero). Therefore, $\mathcal{C}_\text{mom.}$ is a ramified double cover of the elliptic curve $V(f_1)$. 

A curve that is a (ramified) double cover of an elliptic curve is called \emph{bielliptic} \cite{Arbarello1985GeometryAlgebraic} or \emph{elliptic-hyperelliptic} \cite{Accola1994TopicsTheory}. This property is helpful in two ways. First, it can be used to determine the genus once the number of ramification points in $V(f_1,f_2)$ is known; second, it can be used to determine the hyperellipticity by Corollary~\ref{coro} in the Appendix, which says that a bielliptic curve of genus greater than $3$ {\it cannot} be hyperelliptic. 

To count the ramification points, view $\mathcal{C}_\text{mom.}$ as a double cover of $V(f_1)$ via the projection $(t_1,t_2,t_3)\mapsto(t_1,t_3)$. Ramification occurs where $\Delta_{t_2}(f_2)=0$. This discriminant is a degree-$4$ polynomial in $t_3$ with four distinct roots, each produces two ramification points (since $f_1$ is quadratic in $t_1$) as long as $\Delta_{t_1}(f_1)$ and $\Delta_{t_2}(f_2)$ have no common zero, hence $8$ in total, each with ramification index $2$. Using Riemann--Hurwitz formula \eqref{eq:RHformula}, the genus of $\mathcal{C}_\text{mom.}$ is
\begin{equation}
    g(\mathcal{C}_\text{mom.})=\frac{2(2g(V(f_1))-2)+8\times(2-1)+2}{2}=5.
\end{equation}
Since it is greater than $3$, the curve $\mathcal{C}_\text{mom.}$ is not hyperelliptic according to Corollary~\ref{coro}.

\subsection{Loop-by-loop Baikov representation}
In the loop-by-loop Baikov representation, the usual way is to follow the strategies in \cite{Frellesvig2025LoopbyLoopBaikov,Bargiela2025spectrumFeynmanintegral} to derive a loop-by-loop representation that reduces the number of integration variables as much as possible, then take extra residues over irreducible scalar products (ISPs) to further reduce the number of integration variables to one (this can be done by the \verb|LeadingSingularity| command in the \textsc{DlogBasis} \cite{Henn2020Constructingdlog} package). After that one needs to rationalize square roots produced by taking extra residues, which can be done by using the package \textsc{RationalizeRoots} \cite{Besier2020RationalizeRoots}. Then a maximal cut integral is obtained in the form
\begin{equation}\label{eq:LBL-LS}
    \text{MaxCut}[I_{\text{LBL}}^{\text{BPB}}]\sim \int \frac{\dd t}{g(t)\sqrt{f_{8}(t)}},
\end{equation}
where $f_8(t)$ is a degree-8 polynomial of $t$ with no multiple roots, and $g(t)$ is a rational function. Therefore, the curve associated with it is defined by $y^2=f_8(t)$, which is of genus $3$. However, with this approach the expression of the curve becomes very complicated, and is not helpful for seeing the relation to the genus-$5$ curve derived above. 

Therefore we prefer to derive the curve in an alternative way as follows. This approach is still the loop-by-loop representation, except that the last loop momentum is parameterized by spinor helicity to avoid introducing ISPs.

Specifically, we construct the loop-by-loop Baikov kernels only for $l_1$ and $l_2$, and then substitute the parameterization of $l_3$ in Eq.~\eqref{eq:l3} into the remaining scalar products to get
\begin{multline}
    I_{\text{LBL-mom.}}^{\text{BPB}}\sim \int\frac{\dd{z_1}\wedge\dots\wedge\dd{z_6}\wedge\dd{z_{10}}\wedge\dd{z_{11}}}{\sqrt{G(l_1,l_3,p_1,p_2)}\sqrt{G(l_2,l_3,p_4,p_5)}}\\
    \times\frac{\dd{a_{3,1}}\dd{a_{3,2}}\dd{t_3}\dd{a_{3,4}}}{z_1\dots z_{11}},
\end{multline}
then the maximal cut is
\begin{multline}\label{eq:BPB-LBL-mom}
    \text{MaxCut}[I_{\text{LBL-mom.}}^{\text{BPB}}]\sim \\
    \int \frac{\dd t_3}{t_3\sqrt{G(l_1,l_3,p_1,p_2)}\sqrt{G(l_2,l_3,p_4,p_5)}|_\text{MaxCut}},
\end{multline}
where $G(l_1,l_3,p_1,p_2)G(l_2,l_3,p_4,p_5)|_\text{MaxCut}$
is a rational function of $t_3$, whose denominator is a perfect square, and the numerator is a degree-$8$ polynomial with no multiple roots. Therefore, the curve associated with it is again a hyperelliptic curve of genus $3$, and we denote it by $\mathcal{C}_\text{LBL}$. It can be verified to be isomorphic to the curve associated with Eq.~\eqref{eq:LBL-LS} by comparing their absolute Shioda invariants \cite{Shioda1967GradedRing,Shaska2014RemarksHyperelliptic}, using the command \verb|ShaskaInvariants| in the package provided by \cite{Marzucca2024Genusdrop}.

\subsection{The unramified double covering}
The curves derived in the two representations have different genera and therefore cannot be isomorphic. Moreover, the curve $\mathcal{C}_\text{mom.}$ derived in momentum representation is not hyperelliptic, which means it cannot be written in the form $y^2=f(x)$, so it is not easy to use the algorithm provided by \cite{Marzucca2024Genusdrop}, looking for an extra involution to resolve this discrepancy. We will show that an unramified double covering accounts for this genus drop. 

Suppose that there is a holomorphic map between the two curves $h:\mathcal{C}_\text{mom.}\to\mathcal{C}_\text{LBL}$, then it must satisfy the Riemann--Hurwitz formula \eqref{eq:RHformula}:
\begin{equation}
    2g(\mathcal{C}_\text{mom.})-2=n_h\left(2g(\mathcal{C}_\text{LBL})-2\right)+d,
\end{equation}
where $n_{h}>1$ is the degree of $h$, and $d$ is a non-negative integer. With $g(\mathcal{C}_\text{mom.})=5$ and $g(\mathcal{C}_\text{LBL})=3$, the only possible solution for $n_h$ and $d$ is $n_h=2$ and $d=0$,
which means if such a map exists, it must be an {\it unramified} double covering. The existence of such a map will be proven by the construction. 

To see the unramified double covering explicitly, list the analytic expressions of $f_1(t_1,t_3)$ and $G(l_1,l_3,p_1,p_2)$:
\begin{widetext}
\begin{equation}
    \begin{aligned}[b]
        f_1(t_1,t_3)=&\,m^2 t_3^2 \langle 2 3\rangle ^2 [1 2]+t_3 \langle 2 3\rangle
                     \left(-t_1 t_3 \langle 1 3\rangle  [1 2]+m^2 [1 3]\right)s_{12}-m^2 \langle 1 2\rangle  \left(m^2 [1 3]^2+t_1 [1 3] [2
                     3] s_{12}+t_1^2 [2 3]^2 s_{12}\right)\\
                     &+t_1 t_3 s_{12} \left(t_1 t_3 \langle 1
                     3\rangle ^2 [1 2]-t_1 \langle 1 3\rangle  [2 3] s_{12}+m^2
                     s_{13}+\left(m^2-s_{12}\right) s_{23}\right),
    \end{aligned}
\end{equation}
\begin{equation}
    \begin{aligned}[b]
         G(l_1,l_3,p_1,p_2)=&-\frac{1}{16}s_{12}\left[(16 m^2-4s_{12}) (l_3\cdot p_1)^2-4 m^2 s_{12} (l_3\cdot p_1)+3 m^4 s_{12}\right]\\
        =&\frac{s_{12}}{16t_3^2(s_{13}+s_{23})^2} \Big[t_3^4 \langle 1 3\rangle ^2 \langle 2 3\rangle ^2 [1 2]^2\left(-4 m^2+s_{12}\right)+m^4 \langle 1 2\rangle ^2 [1 3]^2 [2 3]^2 \left(-4m^2+s_{12}\right)\\
         &+2 t_3^3 \langle 1 3\rangle  \langle 2 3\rangle  [1 2]s_{12} \left(m^2 s_{13}+\left(-3 m^2+s_{12}\right) s_{23}\right)+2 m^2 t_3 \langle 1 2\rangle  [1 3] [2 3] s_{12} \left(m^2 s_{13}+\left(-3 m^2+s_{12}\right) s_{23}\right)\\
         &+t_3^2 s_{12} \left(-3 m^4 s_{13}^2+2m^4 s_{13} s_{23}+\left(-3 m^4-2 m^2 s_{12}+s_{12}^2\right) s_{23}^2\right)\Big].
    \end{aligned}
\end{equation}
\end{widetext}
Then, it is straightforward to verify that
\begin{equation}\label{eq:DiscG1}
    \frac{\Delta_{t_1}(f_1)}{16t_3^2(s_{13}+s_{23})^2}=G(l_1,l_3,p_1,p_2).
\end{equation}
Note that the denominator of $G(l_1,l_3,p_1,p_2)$ is indeed a perfect square. In the same way, one can also verify that
\begin{equation}\label{eq:DiscG2}
    \frac{\Delta_{t_2}(f_2)}{16t_3^2(s_{35}+s_{34})^2}=G(l_2,l_3,p_4,p_5).
\end{equation}
Therefore, the curve $\mathcal{C}_\text{LBL}$ is actually defined by
\begin{equation}\label{eq:LBLDisc}
    y^2=\Delta_{t_1}(f_1)\Delta_{t_2}(f_2).
\end{equation}
Based on Eq.~\eqref{eq:LBLDisc} and Eq.~\eqref{eq:mom}, one can construct the double covering from $\mathcal{C}_\text{mom.}$ to $\mathcal{C}_\text{LBL}$ explicitly.

Both $f_1$ and $f_2$ are quadratic in their first variable:
\begin{align}
    f_1(t_1,t_3)=A_1(t_3)t_1^2+B_1(t_3)t_1+C_1(t_3),\label{f1}\\
    f_2(t_2,t_3)=A_2(t_3)t_2^2+B_2(t_3)t_2+C_2(t_3),\label{f2}
\end{align}
Introduce the quantities
\begin{align}
    y_1(t_1,t_3)=2A_1(t_3)t_1+B_1(t_3),\label{y1}\\
    y_2(t_2,t_3)=2A_2(t_3)t_2+B_2(t_3),\label{y2}
\end{align}
which satisfy $y_i^2=\Delta_{t_i}(f_i)$ when $f_i(t_i,t_3)=0$. For a fixed $t_3$, denote the two roots of $f_i(t_i,t_3)$ by $t_i^{(+)}$ and $t_i^{(-)}$, then $y_i(t_i^{(+)},t_3)=-y_i(t_i^{(-)},t_3)$. Construct a rational map $\mathcal{C}_\text{mom.} \to \mathcal{C}_\text{LBL}$ on the affine chart as follows:
\begin{equation}\label{eq:bpb-double-covering}
  (t_1,t_2,t_3) \mapsto (t_3,y_1 y_2),
\end{equation}
which extends uniquely to a degree-$2$ holomorphic map between the curves \cite{Hartshorne1977AlgebraicGeometry}, and is unramified by the Riemann--Hurwitz formula.

Why is there an unramified double covering in the momentum representation? The answer turns out to be related to how we apply the cuts on propagators $D_{10}$ and $D_{11}$. To see that, first write down the cut integral in momentum representation after applying cuts on $D_1,\dots,D_9$:
\begin{equation}
    \text{Cut}_{1,\dots,9}[I_\text{mom.}^\text{BPB}]\sim \int \frac{t_3\dd{t_1}\dd{t_2}\dd{t_3}}{f_1(t_1,t_3)f_2(t_2,t_3)}.
\end{equation}
Take the residues with respect to $t_1$ at $f_1(t_1,t_3)=0$ and $t_2$ at $f_2(t_2,t_3)=0$, then we obtain
\begin{equation}
    \text{MaxCut}[I_\text{mom.}^\text{BPB}]\sim \int \frac{t_3\dd{t_3}}{\sqrt{\Delta_{t_1}(f_1)}\sqrt{\Delta_{t_2}(f_2)}}.
\end{equation}
Substituting \eqref{eq:DiscG1} and \eqref{eq:DiscG2} into it, we see that this is the same as Eq.~\eqref{eq:BPB-LBL-mom}, and it involves the genus-$3$ curve $\mathcal{C}_\text{LBL}$. However, if we first take residue with respect to $t_3$ at $f_1(t_1,t_3)=0$, and then $t_2$ at $f_2(t_2,t_3(t_1))=0$, we will get
\begin{equation}
    \int \frac{t_3(t_1)\dd{t_1}}{\sqrt{\Delta_{t_3}(f_1)}\sqrt{\Delta_{t_2}(f_2)}}.
\end{equation}
Note that $\Delta_{t_2}(f_2)$ is a polynomial in $t_3$, but $t_3$ is algebraic in $t_1$ and cannot be rationalized. That is to say, we have eliminated $t_3$ using $f_1(t_1,t_3)=0$ and $f_2(t_2,t_3)=0$, and the result is exactly the intersection of $f_1$ and $f_2$ that generates $\mathcal{C}_\text{mom.}$ in Eq.~\eqref{eq:mom}. Therefore, if one simply uses the intersection of all cut equations as the definition of the curve, one is doing the maximal cut in a way that obscures the simpler geometry, which leads to an unramified double cover with a higher genus.

\subsection{Period relations}
The analytic calculation of Feynman integrals involves the periods of curves found in maximal cuts \cite{Weinzierl2022FeynmanIntegrals}. Periods of compact Riemann surfaces are integrals of holomorphic differentials along contours that cannot be contracted continuously into a point. It turns out that even if we apply the maximal cut in a suboptimal way and get a higher genus, it does not mean that the result is wrong, because the periods of $\mathcal{C}_\text{mom.}$ and $\mathcal{C}_\text{LBL}$ are linearly related as a result of the unramified double covering via the theory of abelian varieties \cite{Birkenhake2004ComplexAbelian}.

For a compact Riemann surface of genus $g$, let $\{\omega_1,\dots,\omega_{g}\}$ be a basis of holomorphic differentials and let $\{\Gamma_1,\dots,\Gamma_{2g}\}$ be a basis of the first homology group, corresponding to the independent contours. Let $\Lambda\subset \mathbb{C}^{g}$ be the lattice spanned by the $2g$ vectors
\begin{equation}
v_i=(\int_{\Gamma_i}\omega_1,\dots,\int_{\Gamma_i}\omega_g).
\end{equation}
The quotient $\mathbb{C}^{g}/\Lambda$ is a complex torus called the Jacobian of the surface. The $g\times 2g$ matrix whose columns are $v_1,\cdots,v_{2g}$ is called the period matrix. Consider an unramified (étale in \cite{Birkenhake2004ComplexAbelian}) double covering $C\to C'$ of compact Riemann surfaces, with Jacobians $\mathrm{Jac}(C)$ and $\mathrm{Jac}(C')$. Then $\mathrm{Jac}(C')$ is an abelian subvariety of $\mathrm{Jac}(C)$, and its complementary abelian subvariety in $\mathrm{Jac}(C)$ is called the Prym variety \cite{Mumford1974PrymVarieties} of the covering $C\to C'$, denoted by $P$. Consequently, $\mathrm{Jac}(C)$ is isogenous to $\mathrm{Jac}(C')\times P$ \cite{Murty1993IntroductionAbelian}.

Applying this to the double covering of Eq.~\eqref{eq:bpb-double-covering}, we conclude that $\mathrm{Jac}(\mathcal{C}_\text{mom.})$ is isogenous to $\mathrm{Jac}(\mathcal{C}_\text{LBL})\times P_\text{BPB}$, where $P_\text{BPB}$ is the associated Prym variety. $P_\text{BPB}$ can be determined using Theorem~\ref{Mumford}.
To this end, we introduce two elliptic curves $E_1$ and $E_2$: the branch points of the elliptic double covering $E_1\to\mathbb{CP}^1$ are the zeros of $\Delta_{t_1}(f_1)$, and those of $E_2\to\mathbb{CP}^1$ are the zeros of $\Delta_{t_2}(f_2)$. Then $\mathcal{C}_\text{mom.}$ is isomorphic to the normalization of $E_1\times_{\mathbb{CP}^1}E_2$. To prove it, it suffices to show that $\mathcal{C}_\text{mom.}$ is birationally equivalent to the normalization of $E_1\times_{\mathbb{CP}^1}E_2$, as both curves are already smooth.

The function field of $\mathcal{C}_\text{mom.}$ is $\mathbb{C}(t_1,t_2,t_3)$ with $f_1(t_1,t_3)=0$ and $f_2(t_2,t_3)=0$; the function field of $E_1\times_{\mathbb{CP}^1}E_2$ is $\mathbb{C}(t_3,y_1,y_2)$ with $y_1^2=\Delta_{t_1}(f_1)$ and $y_2^2=\Delta_{t_2}(f_2)$. Using Eq.~\eqref{f1}--\eqref{y2}, define a rational map $E_1\times_{\mathbb{CP}^1}E_2\to\mathcal{C}_\text{mom.}$ by
\begin{equation}
    (t_3,y_1,y_2)\mapsto\left(\frac{-B_1(t_3)+y_1}{2A_1(t_3)},\frac{-B_2(t_3)+y_2}{2A_2(t_3)},t_3\right),
\end{equation}
and another rational map $\mathcal{C}_\text{mom.}\to E_1\times_{\mathbb{CP}^1}E_2$ by
\begin{equation}
    (t_1,t_2,t_3)\mapsto\left(y_1(t_1,t_3),y_2(t_2,t_3),t_3\right).
\end{equation}
It is then straightforward to verify that both maps are well-defined and inverses of each other except at a finite set of points. Consequently, $\mathcal{C}_\text{mom.}$ and the normalization of $E_1\times_{\mathbb{CP}^1}E_2$ are birationally equivalent, hence isomorphic. According to Theorem~\ref{Mumford}, we have $P_\text{BPB}\cong\mathrm{Jac}(E_1)\times\mathrm{Jac}(E_2)$. As $\mathrm{Jac}(\mathcal{C}_\text{mom.})$ is isogenous to $\mathrm{Jac}(\mathcal{C}_\text{LBL})\times P_\text{BPB}$ by definition, we conclude that $\mathrm{Jac}(\mathcal{C}_\text{mom.})$ is isogenous to $\mathrm{Jac}(E_1)\times\mathrm{Jac}(E_2)\times\mathrm{Jac}(\mathcal{C}_\text{LBL})$. Therefore, after suitable choices of bases for holomorphic differentials and homology cycles, the period matrices satisfy a linear relation as follows \cite{Birkenhake2004ComplexAbelian}: we denote the period matrices of $\mathcal{C}_\text{mom.},\mathcal{C}_\text{LBL},E_1$ and $E_2$ by $\mathcal{P}_\text{mom.}\in\mathbb{C}^{5\times 10},\mathcal{P}_\text{LBL}\in\mathbb{C}^{3\times 6},\mathcal{P}_1\in\mathbb{C}^{1\times 2}$ and $\mathcal{P}_2\in\mathbb{C}^{1\times 2}$ respectively, then there exist matrices $M\in \mathbb{C}^{5\times 5}$ and $N\in\mathbb{Z}^{10\times 10}$ such that
\begin{equation}\label{eq:periods-BPB}
    M\mathcal{P}_\text{mom.}=
    \begin{pmatrix}
        \mathcal{P}_1 & 0 & 0\\
        0 & \mathcal{P}_2 & 0\\
        0 & 0 & \mathcal{P}_\text{LBL}
    \end{pmatrix}
    N.
\end{equation}

\subsection{More three-loop examples}
The above genus drop mechanism exists in some other three-loop Feynman diagrams with internal masses as well. The genera can be derived in a similar way as before. We summarize them with the box-pentagon-box diagram in Table~\ref{table:genera}, which all belong to the case Type I. The genera in momentum representation agree with the results of \cite{Huang2013GeneraCurves,Hauenstein2015GlobalStructure}, and in this representation these curves are non-hyperelliptic.

\begin{table}[h!]
\centering
\begin{tblr}{colspec={|Q[c,m]|Q[c,m]|Q[c,m]|Q[c,m]|},colsep=2pt}
\hline
&
\begin{tikzpicture}[x=0.75pt,y=0.75pt,yscale=-1,xscale=1,scale=0.35,baseline={([yshift=-3pt]current bounding box.center)}]
\draw  [line width=1.0]  (278.45,74.79) -- (264.79,116.02) -- (221.36,115.77) -- (208.18,74.39) -- (243.46,49.06) -- cycle ;
\draw [line width=0.5]    (130.33,85.33) -- (169.46,86.68) ;
\draw [line width=0.5]    (182.59,128.02) -- (146.33,150.33) ;
\draw [line width=0.5]    (317.1,87.28) -- (360.33,93.33) ;
\draw [line width=0.5]    (303.77,128.55) -- (341.33,153.33) ;
\draw [line width=0.5]    (243.46,49.06) -- (243.33,16.33) ;
\draw [line width=1.0]    (169.46,86.68) -- (182.59,128.02) ;
\draw [line width=1.0]    (169.46,86.68) -- (208.18,74.39) ;
\draw [line width=1.0]    (317.1,87.28) -- (303.77,128.55) ;
\draw [line width=1.0]    (317.1,87.28) -- (278.45,74.79) ;
\draw [line width=1.0]    (303.77,128.55) -- (264.79,116.02) ;
\draw [line width=1.0]    (182.59,128.02) -- (221.36,115.77) ;
\end{tikzpicture}
&
\begin{tikzpicture}[x=0.75pt,y=0.75pt,yscale=-1,xscale=1,scale=0.4,baseline={([yshift=-3pt]current bounding box.center)}]
\draw [line width=1.0]    (230,110) -- (230.33,170.67) ;
\draw [line width=1.0]    (230.33,170.67) -- (291.33,170.67) ;
\draw [line width=1.0]    (188.33,110.67) -- (230,110) ;
\draw [line width=1.0]    (188.67,171.33) -- (230.33,170.67) ;
\draw [line width=1.0]    (188.33,110.67) -- (188.67,171.33) ;
\draw [line width=1.0]    (273.33,139.67) -- (291.33,170.67) ;
\draw [line width=1.0]    (230,110) -- (260.33,79.67) ;
\draw [line width=1.0]    (260.33,79.67) -- (291,110) ;
\draw [line width=1.0]    (273.33,139.67) -- (291,110) ;
\draw [line width=1.0]    (291.33,170.67) -- (309,141) ;
\draw [line width=1.0]    (291,110) -- (309,141) ;
\draw [line width=0.5]    (247.33,139.67) -- (273.33,139.67) ;
\draw [line width=0.5]    (160.33,94.67) -- (188.33,110.67) ;
\draw [line width=0.5]    (161.33,182.67) -- (188.67,171.33) ;
\draw [line width=0.5]    (260.33,52.67) -- (260.33,79.67) ;
\draw [line width=0.5]    (309,141) -- (335,141) ;
\end{tikzpicture}
&
\begin{tikzpicture}[x=0.75pt,y=0.75pt,yscale=-1,xscale=1,scale=0.4,baseline={([yshift=-3pt]current bounding box.center)}]
\draw [line width=1.0]    (250.33,190.67) -- (311.33,190.67) ;
\draw [line width=1.0]    (293.33,159.67) -- (311.33,190.67) ;
\draw [line width=1.0]    (250,130) -- (280.33,99.67) ;
\draw [line width=1.0]    (280.33,99.67) -- (311,130) ;
\draw [line width=1.0]    (293.33,159.67) -- (311,130) ;
\draw [line width=1.0]    (311.33,190.67) -- (329,161) ;
\draw [line width=1.0]    (311,130) -- (329,161) ;
\draw [line width=0.5]    (280.33,72.67) -- (280.33,99.67) ;
\draw [line width=0.5]    (329,161) -- (346.33,143) ;
\draw [line width=0.5]    (276,177.67) -- (293.33,159.67) ;
\draw [line width=1.0]    (232.33,159.67) -- (250,130) ;
\draw [line width=1.0]    (232.33,159.67) -- (250.33,190.67) ;
\draw [line width=1.0]    (250,130) -- (268,161) ;
\draw [line width=0.5]    (268,161) -- (285.33,143) ;
\draw [line width=1.0]    (250.33,190.67) -- (268,161) ;
\draw [line width=0.5]    (215,177.67) -- (232.33,159.67) ;
\end{tikzpicture}
\\
\hline
momentum & 5 & 9 & 13 \\
\hline
loop-by-loop & 3 & 5 & 7 \\
\hline
\end{tblr}
\caption{Genera of some three-loop diagrams}
\label{table:genera}
\end{table}

\section{Type II}
We now analyze the genus drop between hyperelliptic curves, taking the non-planar crossed box \cite{Georgoudis2015TwoloopIntegral,Marzucca2024Genusdrop} as an example, shown in Fig.~\ref{fig:NPCB}. By expressing the ISPs in terms of the spherical coordinates in the Baikov representation, we will make manifest the extra involution found by \cite{Marzucca2024Genusdrop} and its relation to the discrete Lorentz symmetries.

\begin{figure}[h]
\begin{tikzpicture}[x=0.75pt,y=0.75pt,yscale=-1,xscale=1]
\draw  [line width=1.5]  (244.48,102.29) -- (279.84,137.65) -- (244.48,173) -- (209.13,137.65) -- cycle ;
\draw    (119,69) -- (154.33,103) ;
\draw    (155.33,173) -- (121.33,203) ;
\draw    (178.33,137) -- (209.65,137.52) ;
\draw    (279.84,137.65) -- (319.82,138.16) ;
\draw [line width=1.5]    (154.33,103) -- (244.48,102.29) ;
\draw [line width=1.5]    (154.33,103) -- (155.33,173) ;
\draw [line width=1.5]    (155.33,173) -- (244.48,173) ;
\draw (194,181.4) node [anchor=north west][inner sep=0.75pt]    {$l_{1}$};
\draw (265,160.4) node [anchor=north west][inner sep=0.75pt]    {$l_{2}$};
\draw (140.33,191.4) node [anchor=north west][inner sep=0.75pt]    {$p_{1}$};
\draw (140,66.4) node [anchor=north west][inner sep=0.75pt]    {$p_{2}$};
\draw (301.83,141.3) node [anchor=north west][inner sep=0.75pt]    {$p_{3}$};
\draw (180.33,140.4) node [anchor=north west][inner sep=0.75pt]    {$p_{4}$};
\draw (134,130.4) node [anchor=north west][inner sep=0.75pt]    {$m$};
\end{tikzpicture}
\caption{The non-planar crossed box in \cite{Marzucca2024Genusdrop}}
\label{fig:NPCB}
\end{figure}
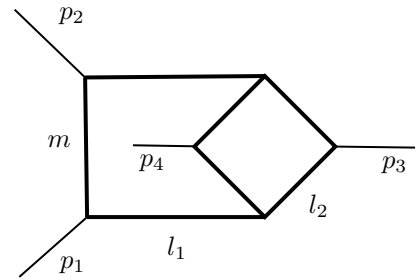

In an approach similar to \cite{Mastrolia2016AdaptiveIntegrand}, the loop-by-loop Baikov representation can be written in the form
\begin{equation}\label{eq:NPCB}
    I^{\text{NPCB}}_\text{angular}\sim\int \frac{\dd{z_1}\wedge\dots\wedge\dd{z_7}}{\sqrt{G(l_2,p_3,p_4,l_1)}}\frac{\dd{\theta}}{z_1\dots z_7}.
\end{equation}
To derive it, we first recall the derivation of the Baikov representation \cite{Frellesvig2025LoopbyLoopBaikov}. For each loop, the loop momentum $l$ is split to a parallel space component $l^{\parallel}$ and a perpendicular space component $l^{\perp}$. The parallel space is spanned by the external momenta, denoted by $\{q_1,\dots,q_E\}$. Then, the integration of the parallel component is transformed to the scalar products $\{l\cdot q_1,\dots,l\cdot q_{E}\}$. In the perpendicular space, the integration is over the norm $|l^{\perp}|=\sqrt{-(l^{\perp})^2}$ and the spherical coordinates $\Omega_{D-E-1}$. Usually, the parallel space should include all the ISPs, so that the spherical coordinates can be integrated to a constant. In this way, the Jacobian (also called the Baikov kernel) is
\begin{equation}
G(l,q_1,\dots,q_{E})^{\frac{D-E-2}{2}}G(q_1,\dots,q_{E})^{-\frac{D-E-1}{2}}.
\end{equation}
Applying this to the $l_2$ loop in Fig.~\ref{fig:NPCB}, we get the kernel $G(l_2,p_3,p_4,l_1)^{-1/2}$ in \eqref{eq:NPCB} under $D\to 4$.

However, here we introduce a trick to leave only the propagators in the parallel space and use spherical coordinates to represent the ISPs. For example, in the non-planar crossed box diagram, the ISP associated with $l_1$ is just $l_1\cdot p_4$. We can set $l_{1}^{\parallel}$ to lie in the span of $\{p_1,p_2\}$, then the spherical coordinate is one-dimensional, denoted by $\theta$. Then $l_1\cdot p_4$ can be represented by
\begin{equation}
    \begin{aligned}[b]
        l_1\cdot p_4=&l_{1}^{\parallel}\cdot p_{4}^{\parallel}+l_{1}^{\perp}\cdot p_{4}^{\perp}\\
        =&c_1 l_{1}^{\parallel}\cdot p_1+c_2  l_{1}^{\parallel}\cdot p_2+l_{1}^{\perp}\cdot p_{4}^{\perp}\\
        =&c_1 l_1 \cdot p_1+c_2 l_1 \cdot p_2+|l_{1}^{\perp}|\times|p_{4}^{\perp}|\cos\theta,
    \end{aligned}
\end{equation}
where $p_{4}^{\parallel}=c_1 p_1+c_2 p_2$. In this case, the exponent of the Baikov kernel $G(l_1,p_1,p_2)$ goes to zero under $D\to 4$, and the integration variables are $z_5,z_6,z_7$ and $\theta$. This is how we derive Eq.~\eqref{eq:NPCB}.

Applying the maximal cut to Eq.~\eqref{eq:NPCB}, we obtain
\begin{equation}\label{eq:NPCB-G}
    \begin{aligned}[b]
        G(l_2,p_3,p_4,l_1)=&-\frac{1}{16} \left[3 m^2 s-2 s (l_1\cdot p_4)-4 \left(l_1\cdot p_4\right)^2\right]\\
        &\times\left[m^2 s+2 s (l_1\cdot p_4)+4 \left(l_1\cdot p_4\right)^2\right],
    \end{aligned}
\end{equation}
and
\begin{equation}\label{eq:l1p4}
    \begin{aligned}[b]
        l_1\cdot p_4=&\frac{t}{2}+\sqrt{(l_{1}^{\perp})^2 (p_{4}^{\perp})^2}(-\cos\theta)\\
        =&\frac{t}{2}-\sqrt{\frac{m^2 t (s+t)}{s}}\cos\theta,
    \end{aligned}
\end{equation}
where $s$ and $t$ are Mandelstam variables. Note that the minus sign with $\cos\theta$ comes from the fact that the metric in the perpendicular space has the negative signature (a vector that is orthogonal to a lightlike vector is either spacelike or proportional to it, thus the subspace orthogonal to a lightlike vector contains no timelike directions). With Eq.~\eqref{eq:l1p4} and Eq.~\eqref{eq:NPCB-G}, and parameterizing $\cos\theta$ by the classical trigonometric parameterization $\cos\theta\mapsto(1-x^2)/(1+x^2)$, we obtain
\begin{equation}
    G(l_2,p_3,p_4,l_1)=\frac{P_8(x)}{16 s^2 \left(x^2+1\right)^4}
\end{equation}
Since the denominator is a perfect square, the curve associated with it is determined by the numerator, $P_8(x)$. This polynomial is of degree $8$ in $x$ with no multiple roots; thus, the curve associated with it is of genus $3$, and is isomorphic to the curve found in \cite{Georgoudis2015TwoloopIntegral} (after setting the masses as the same), again verified by computing the Shioda invariants. Because $x$ comes from the parameterization for $\cos \theta$, it only appears in the form $x^2$, and naturally possesses an extra involution. Now we can transform the integration variable from $x$ to $w=-x^2$,  to obtain
\begin{equation}
    \text{MaxCut}[I^{\text{NPCB}}_\text{angular}]\sim\int\frac{1-w}{\sqrt{P_5(w)}}\dd{w},
\end{equation}
where
$P_5(w)=-wP_8(\sqrt{-w})$ is a degree-$5$ polynomial in $w$ with no multiple roots.
This is exactly the genus-dropped $Q_5$ found in \cite{Marzucca2024Genusdrop} by searching for extra involution, which represents a genus-2 hyperelliptic curve.

There is also an unramified double covering from curve $H:y^2=P_8(x)$ to curve $H':u^2=P_5(w)$, given by the unique extension \cite{Hartshorne1977AlgebraicGeometry} of the following rational map:
\begin{equation}
  (x,y)\mapsto (w=-x^2,u=xy),
\end{equation}
which is akin to the $\rho_2$ in \cite{Marzucca2024Genusdrop}. 
Split the six branch points (one is at infinity) of $y^2=P_5(w)$ into two disjoint subsets denoted by $V_1=\{0,\infty\}$ and $V_2=\text{zeros of }R_4(w)$ where $R_4(w)=-P_5(w)/w$. Associated with them, define two double covers of $\mathbb{CP}^1$, denoted by $F_1:y_1^2=-w$ and $F_2:y_2^2=R_4(w)$, then $H:y^2=P_8(x)$ is isomorphic to the normalization of $F_1\times_{\mathbb{CP}^1}F_2$. Again verifying the birational equivalence is sufficient to prove it. Define the rational maps as follows:
\begin{align}
    (x,y)&\mapsto (w=-x^2,y_1=x,y_2=y),\\
    (w,y_1,y_2)&\mapsto (x=y_1,y=y_2),
\end{align}
which are inverses of each other. Therefore, $\mathrm{Jac}(F_2)$ is the Prym variety associated with the unramified double covering $H\to H'$ according to Theorem~\ref{Mumford}. Then there is an isogeny between $\mathrm{Jac}(H)$ and $\mathrm{Jac}(F_2)\times\mathrm{Jac}(H')$, and a period matrix relation similar to Eq.~\eqref{eq:periods-BPB}.

This genus drop comes from the fact that there is only one ISP $l_1\cdot p_4$, so there is no need to introduce $\sin\theta$, which makes the trigonometric function parameter $x$ only appear as $x^2$. This is a consequence of the external momenta spanning only a three-dimensional subspace. Indeed, there exists a reflection symmetry
\begin{equation}
    \mathcal{R}: v^\parallel \mapsto v^\parallel,\quad v^\perp \mapsto -v^\perp,
\end{equation}
where $v^\parallel \in \operatorname{span}(p_1,p_2,p_3,p_4)$ and $v^\perp$ lies in the orthogonal direction. 
$\mathcal{R}$ sends the spherical coordinate $\theta\mapsto -\theta$, while $l_1\cdot p_4$ remains invariant. This is the physical origin of the extra involution $x\mapsto -x$ that drives the genus drop. For generic kinematics with external momenta spanning the full four-dimensional spacetime, this particular reflection symmetry is absent and there is no genus drop, as observed in \cite{Bargiela2025spectrumFeynmanintegral}. This also provides an explanation from a physical point of view for the observation in \cite{Hauenstein2015GlobalStructure} that no even genus was found in three-loop diagrams. Since unramified double coverings always produce genus drops as $g\mapsto (g+1)/2$, the generic genus must be odd for those diagrams that can undergo further genus drop in special kinematics where the external legs span only a subspace.

\section{A non-hyperelliptic diagram with no known genus drop}
We point out that there are Feynman integrals that involve non-hyperelliptic curves for which no genus drop has been found. The massive pentagon-box-pentagon diagram shown in Fig.~\ref{fig:PBP} serves as an example.

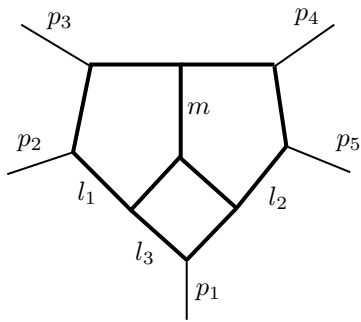
\begin{figure}[ht]
\begin{tikzpicture}[x=0.75pt,y=0.75pt,yscale=-1,xscale=1]
\draw [line width=1.5]    (294.33,170) -- (322.33,195) ;
\draw  [line width=1.5]  (274.33,97) -- (319.33,97) -- (319.33,143) -- (294.33,170) -- (265.33,141) -- cycle ;
\draw [line width=1.5]    (319.33,97) -- (366.33,97) ;
\draw [line width=1.5]    (366.33,98) -- (372.33,138) ;
\draw [line width=1.5]    (319.33,144) -- (347.33,169) ;
\draw [line width=1.5]    (322.33,195) -- (347.33,169) ;
\draw [line width=1.5]    (347.33,169) -- (372.33,138) ;
\draw    (232.33,152) -- (265.33,141) ;
\draw    (274.33,98) -- (239.33,77) ;
\draw    (366.33,98) -- (396.33,77) ;
\draw    (372.33,138) -- (404.33,151) ;
\draw    (322.33,195) -- (322.33,225) ;
\draw (321,114.4) node [anchor=north west][inner sep=0.75pt]    {$m$};
\draw (266,154.4) node [anchor=north west][inner sep=0.75pt]    {$l_{1}$};
\draw (361.83,156.9) node [anchor=north west][inner sep=0.75pt]    {$l_{2}$};
\draw (295,184.4) node [anchor=north west][inner sep=0.75pt]    {$l_{3}$};
\draw (325.33,205.4) node [anchor=north west][inner sep=0.75pt]    {$p_{1}$};
\draw (236,130.4) node [anchor=north west][inner sep=0.75pt]    {$p_{2}$};
\draw (251,69.4) node [anchor=north west][inner sep=0.75pt]    {$p_{3}$};
\draw (375,67.4) node [anchor=north west][inner sep=0.75pt]    {$p_{4}$};
\draw (396,130.4) node [anchor=north west][inner sep=0.75pt]    {$p_{5}$};
\end{tikzpicture}
\caption{The pentagon-box-pentagon diagram with massive internal lines}
\label{fig:PBP}
\end{figure}

We use a parameterization similar to Eq.~\eqref{eq:BPB-LBL-mom}. We construct the Baikov kernel only for $l_3$ and parameterize $l_1$ and $l_2$ using spinor helicity to represent ISPs. Then we apply cuts on all propagators except $(l_1+l_2+p_1)^2-m^2$ to obtain
\begin{multline}
    \text{SubMaxCut}[I^{\text{PBP}}_\text{LBL}]\\
    \sim\int \frac{\dd{t_1} \dd{t_2}}{((l_1+l_2+p_1)^2-m^2)\sqrt{G(l_3,p_1,l_1,l_2)}}.
\end{multline}
Here, $t_1$ and $t_2$ are the remaining loop momentum parameters. We do not apply the maximal cut directly because the last cut equation $(l_1+l_2+p_1)^2-m^2=0$ itself produces an elliptic curve. As a result, the curve associated with this diagram is given by 
\begin{equation}\label{eq:curve-PBP}
    \begin{cases}
        y^2=G(l_3,p_1,l_1,l_2),\\
        (l_1+l_2+p_1)^2-m^2=0,
    \end{cases}
\end{equation}
with the sub-maximal cut implicitly applied. This curve is bielliptic because it is a ramified double cover of the elliptic curve $(l_1+l_2+p_1)^2-m^2=0$. Solving $G(l_3,p_1,l_1,l_2)=0$ and $(l_1+l_2+p_1)^2-m^2=0$ simultaneously for $(t_1,t_2)$ gives $16$ ramification points of ramification index $2$, including the possible points at infinity. Using the Riemann--Hurwitz formula \eqref{eq:RHformula}, we conclude that this curve is of genus $9$. According to Corollary~\ref{coro}, it is non-hyperelliptic. Unlike previous examples, we have not found a genus drop pattern to simplify it to a hyperelliptic curve. It is therefore a candidate for a genuinely non-hyperelliptic Feynman integral.

\section{Conclusion and outlook}
We have shown that both types of genus drop encountered in Feynman integrals (I) the one accompanied by a curve-type change from non-hyperelliptic to hyperelliptic, and (II) the one occurring between hyperelliptic curves---are tied to the same mathematical structure: an unramified double covering between the algebraic curves determined from maximal cuts in different representations. The period matrices of the higher-genus and lower-genus curves are linearly related via the Prym varieties associated with the unramified double coverings. The physical origins behind the two types are distinct. In the case (I), the genus drop arises from the order in which residues are taken when imposing the maximal cut. In the case (II), the genus drop is driven by a reflection symmetry with respect to the spacetime dimension orthogonal to the external momenta, which arises because the external legs span only a three-dimensional subspace.

Despite this progress, an efficient method for rapidly determining the geometry associated with an arbitrary high-loop Feynman integral is still on demand. The existence of the massive pentagon-box-pentagon diagram, whose maximal cut yields a genus-9 non-hyperelliptic curve with no apparent genus drop, illustrates that the there could be genuinely non-hyperelliptic Feynman integrals. Developing a systematic approach to extract the minimal geometric structure associated with a high-loop Feynman integral is therefore an important direction for future research.

\section*{Acknowledgments}
We acknowledge Hjalte Frellesvig, Xing Wang, Stefan Weinzierl and Yu Wu, for enlightening discussions.  JYG is supported by the National Natural Science Foundation of
China under Grant No. 12357077.
 YZ is supported by NSFC through Grant No. 12575078 and 12247103,
and would like to thank the Erwin Schr\"odinger International Institute
for Mathematics and Physics (ESI), University of Vienna (Austria), for the opportunity to participate
in the Thematic Programme “Amplitudes and Algebraic Geometry” in 2026 where a significant part
of this work has been accomplished and for the support given.

\appendix*
\section{}
\begin{theorem}[Riemann--Hurwitz, see e.g. \cite{Farkas1992RiemannSurfaces}]\label{RH}
    Let $f:X\to Y$ be a covering map of degree $n$ between compact Riemann surfaces of genera $g(X)$ and $g(Y)$ respectively. Suppose there are $m$ points in $X$ with ramification indices greater than $1$, denoted by $r_1,r_2,\dots,r_m$ respectively. Then
    \begin{equation}\label{eq:RHformula}
        2g(X)-2=n(2g(Y)-2)+\sum_{i=1}^{m}(r_i-1).
    \end{equation}
\end{theorem}

\begin{theorem}[Castelnuovo--Severi, \cite{Castelnuovo1906Sulleserie,Accola1994TopicsTheory}]
    Suppose there are three compact Riemann surfaces $X_0,X_1,X_2$ of genera $g_0,g_1,g_2$ respectively, with two holomorphic maps $f_1:X_0\to X_1$ and $f_2:X_0\to X_2$ of degrees $n_1$ and $n_2$ respectively. Suppose further that there does not exist a compact Riemann surface $X'$ of genus $g' < g_0$ with a holomorphic map $f':X_0\to X'$ such that $f_1$ and $f_2$ factorize through it, i.e. $f_1=f_1'\circ f'$ for some $f_1':X'\to X_1$ and $f_2=f_2'\circ f'$ for some $f_2':X'\to X_2$. Then
    \begin{equation}\label{eq:CS}
        g_0\leq n_1 g_1+n_2 g_2+(n_1-1)(n_2-1).
    \end{equation}
\end{theorem}

\begin{corollary}\label{coro}
    A compact Riemann surface of genus greater than $3$ cannot be both bielliptic and hyperelliptic.
    \begin{proof}
        Let $C$ be a bielliptic and hyperelliptic curve of genus $g(C)$; then, there are two degree-$2$ covering maps $f_1:C\to\mathbb{C}\mathbb{P}^1$ and $f_2:C\to E$, where $E$ is an elliptic curve. Assume that $f_1$ and $f_2$ factorize through $h:C\to X$, where $X$ is a compact Riemann surface of genus $g(X)<g(C)$. Then there exist $f_1':X\to \mathbb{CP}^1$ and $f_2':X\to E$ such that $f_1=f_1'\circ h$ and $f_2=f_2'\circ h$, and we have $\deg f_1=\deg f_1'\cdot\deg h$ and $\deg f_2=\deg f_2'\cdot\deg h$. Since $g(X)<g(C)$, $\deg h>1$. Moreover, $\deg h\leq\deg f_1=\deg f_2=2$, so we have $\deg h=2$, and therefore $\deg f_1'=\deg f_2'=1$, which means that both $\mathbb{CP}^1$ and $E$ are isomorphic to $X$, which is a contradiction. Therefore, according to the Castelnuovo--Severi inequality \eqref{eq:CS}, $g(C)\leq2g(\mathbb{CP}^1)+2g(E)+1=3$.
    \end{proof}
\end{corollary}

\begin{theorem}[\cite{Mumford1974PrymVarieties}]\label{Mumford}
    Let $C$ be a hyperelliptic curve of genus $g$, and $\pi:C\to\mathbb{CP}^1$ be the corresponding hyperelliptic double covering with branch points $w_1,\dots,w_{2g+2}$. Separate the branch points into two groups of even cardinality: $\{w_1,\dots,w_{2g+2}\}=Y_1\cup Y_2$, with $2k_1+2$ elements in $Y_1$, $2k_2+2$ elements in $Y_2$, and $Y_1\cap Y_2=\varnothing$, hence $k_1+k_2+1=g$. Let $\pi_1:C_1\to\mathbb{CP}^1$ and $\pi_2:C_2\to\mathbb{CP}^1$ be the double coverings with branch points $Y_1$ and $Y_2$ respectively. Let $\widetilde{C}$ be the normalization of the fiber product $C_1\times_{\mathbb{CP}^1}C_2$, then $\widetilde{C}$ is an unramified double cover of $C$, and the associated Prym variety $\text{Prym}(\widetilde{C}/C)$ is isomorphic to $\mathrm{Jac}(C_1)\times\mathrm{Jac}(C_2)$. If one of $k_1$ and $k_2$ is zero, then the corresponding factor disappears.
\end{theorem}

\bibliographystyle{apsrev4-2}
\bibliography{references}
\end{document}